\pdfoutput=1
\documentclass[10pt,twocolumn,conference]{IEEEtran}
\usepackage{amsmath,amsfonts,amssymb}
\usepackage{graphicx}
\usepackage{cite}
\usepackage{multirow}
\usepackage{algorithm}
\usepackage{algorithmic}
\usepackage{subfigure}

\IEEEoverridecommandlockouts

\newtheorem{theorem}{Theorem}
\newtheorem{definition}{Definition}

\newtheorem{corollary}{Corollary}

\newtheorem{example}{Example}

\begin{document}
\title{Reduced Functional Dependence Graphs and Their Applications}
\author{\authorblockN{Xiaoli~Xu\dag, Satyajit Thakor\ddag \ and Yong~Liang~Guan\dag\\}
\authorblockA{\dag School of Electrical and Electronic Engineering, Nanyang Technological University, Singapore\\
\ddag Institute of Network Coding, The Chinese University of Hong Kong, Shatin, Hong Kong\\
\{xu0002li, eylguan\}@e.ntu.edu.sg, thakor@inc.cuhk.edu.hk}}
\maketitle
\begin{abstract}
Functional dependence graphs (FDGs) are an important class of directed graphs that capture the functional dependence relationship among a set of random variables. FDGs are frequently used in characterizing and calculating network coding capacity bounds. However, the order of an FDG is usually much larger than the original network and the complexity of computing bounds grows exponentially with the order of an FDG. In this paper, we introduce graph pre-processing techniques which deliver reduced FDGs. These reduced FDGs are obtained from the original FDG by removing nodes that are not ``essential''. We show that the reduced FDGs give the same capacity region/bounds obtained using original FDGs, but require much less computation. The application of reduced FDGs for algebraic formulation of scalar linear network coding is also discussed.
\end{abstract}


\section{Introduction}
Characterizing the capacity region of network coding is an important fundamental problem. For single session networks, it is well known that the capacity region is given by the max-flow bound \cite{AhlCai00}. However, the max-flow bound is no longer tight when multi-session networks are considered. It is shown in \cite{Yan07} that the exact capacity region for general networks can be written as a function of entropy region, $\Gamma_{n}^{\star}$, and some constraints induced by the network topology. Unfortunately, characterization of $\Gamma_{n}^{\star}$ is still open for $n\geq4$ and infinite number of information inequalities need to be considered \cite{Matus07}. An explicit outer bound in terms of the set of polymatroidal functions, $\Gamma_{n}$ is referred as Linear Programming (LP) bound \cite{Yeung08}.

Functional dependence graphs (FDG) are a class of directed graphs that capture the functional dependence relationship among a set of random variables. FDGs are first used in \cite{Kra98} to establish conditional independence among random variables involved in a communication system which is useful to characterize bounds on the capacity of the communication system. Variants of FDGs are used in \cite{{KraSav06}} and \cite{Thakor09} to characterize computable outer bounds on multi-session network coding capacity.
%The concept of information dominance in the network is firstly introduced in \cite{Harvey06} in order to obtain a tighter capacity bound than the min-cut bound. Functional dependence graph (FDG) that captures the information dominance relationship in the network is introduced in \cite{Thakor09} and used to obtain a functional dependence outer bound for network coding.
FDG plays an important role in studying the capacity for networks with edge capacity constraints. Besides the progressive \emph{d}-separating edge-set bound \cite{{KraSav06}} and the functional dependence bound \cite{Thakor09} which are obtained using FDGs, many other capacity bounds e.g. \cite{Harvey06}, and the LP bound are closely related with FDGs of the network.  The random variables and the functional dependence constraints which are important in computing the LP bound of a given network are directly reflected by FDGs. % and the dimension of the problem is equal to $2^{\{order\  of\  FDG\}}$.

In this paper, we introduce graph pre-processing techniques to reduce the size of an FDG by removing edges and nodes that are not essential. The resulting FDG is refereed as Reduced Functional Dependence Graph (FDG) which is of smaller size but preserves important properties of the original network. We show that the capacity of the original network can totally be determined by the reduced FDGs. As a result, reduced FDGs can be used to compute bounds on the network coding capacity of a given network with lesser computational resources. Moreover, removing nodes in the original FDGs is equivalent to removing the random variables involved in computing the LP bound of a given network. Hence the complexity of computing LP bound reduces exponentially using the notion of reduced FDGs.

%on We prove that the capacity of original network is totally determined by the reduced FDG.

Due to the similarities between an FDG and the directed line graph \cite{Koetter03} of a given network, an FDG can also be used to construct the algebraic formulation of scalar linear network coding. With the proposed reduced FDG, the number of variables required for such formulation will be reduced and the complexity for computing the transition matrix can also be significantly reduced.

The rest of the paper is organized as follows. In section II, network model and formal definition of an FDG are given. Graph pre-processing techniques leading to reduced FDGs are given in section III. Applications of the main results are discussed in section IV, including complexity reduction for computing LP bound and complexity reduction for algebraic formulation of scalar linear network coding. Finally, the paper is concluded in section V.

%In the sequel, set elements are denoted by numbers or capital letters e.g., $1,A$, sets are denoted by font e.g., $\mathcal A$. The power set of a set $\mathcal A$ is denoted by $\boldsymbol{\mathcal P}(\mathcal A)$.% The notation $A_{\mathcal B}$ means the set $\{A_{B}:B \in \mathcal B\}$.

\section{Background}\label{S:background}
\subsection{Network Model}
A network is represented by a directed acyclic graph $\mathcal G=(\mathcal V,\mathcal E)$, where $\mathcal V$ is the set of nodes and $\mathcal E$ is the set of edges.  For an edge $E=(V,U)\in{\mathcal E}$, define $\mathrm{Tail}(E)=V$ and $\mathrm{Head}(E)=U$. For $V\in \mathcal V$, the set of edges entering into $V$ and leaving $V$ are denoted by $\mathrm{In}(V)$ and $\mathrm{Out}(V)$ respectively. To simplify the description, we also define the set of edges entering into and leaving an edge $E$ as $\mathrm{In}(E)=\{E':\mathrm{Head}(E')=\mathrm{Tail}(E)\}$ and $\mathrm{Out}(E)=\{E':\mathrm{Tail}(E')=\mathrm{Head}(E)\}$ respectively. It is easy to see that if $E=(V,U)$, $\mathrm{In}(E)=\mathrm{In}(V)$ and $\mathrm{Out}(E)=\mathrm{Out}(U)$.

Let $\mathcal S=\{1,\cdots,|\mathcal S|\}$ denote the set of independent information sources available at some nodes (called source nodes) in a network via mapping $a:\mathcal S \mapsto \mathcal V$. The sources are demanded by some nodes in a network called sink nodes. A set of sources demanded by a given sink node is described by the mapping $b:\mathcal V \mapsto \boldsymbol{\mathcal P}(\mathcal S)$ where, $\boldsymbol{\mathcal P}(\mathcal A)$ is the power set of the set $\mathcal A$. Let $\mathcal T$ denote the set of sink nodes in the network and thus the set of sources demanded by the sink $T\in\mathcal T$ is $b(T)$. If each source is demanded by exactly one sink, the network is called multiple-unicast network. Without loss of generality, we assume that $\mathrm{In}(S)=\mathrm{Out}(T)=\emptyset, \forall S\in\mathcal S,T\in\mathcal T$. %This can be achieved by adding imaginary source nodes and sink nodes that are connected to the original source and sink nodes via a link with infinite capacity. For $E\in\mathcal E$,  $C_E$ denotes the maximal rate that can be conveyed through the link $E$.
For $E\in\mathcal E$,  $C_E$ denotes the maximal rate that can be conveyed through the link $E$.

\begin{definition}\label{NC}
%Let $R_S$ denote the rate of information source $S\in\mathcal S$.
For a given network $\mathcal G=(\mathcal V,\mathcal E)$, an information rate tuple $\mathbf{R}=(R_S:S\in\mathcal S)$ is achievable if there exists a network code $\textbf{\textsf{C}}=\{f_E: E \in \mathcal E, g_T: T \in \mathcal T\}$ of block length $n$, defined by
\begin{itemize}
\item{For all $S\in\mathcal S, E\in{\mathrm{Out}(S)}$, local encoding functions $f_E: \mathcal Y_S^{(n)} \rightarrow \mathcal U_E^{(n)},$}
\item{For all $V\in\mathcal V\setminus(\mathcal S\cup\mathcal T), E\in \mathrm{Out}(V)$ local encoding functions $f_E: \prod_{E'\in \mathrm{In}(V)}\mathcal U_{E'}^{(n)}\rightarrow \mathcal U_{E}^{(n)},$}
\item{and for all $T\in\mathcal T$, decoding functions $g_T:\prod_{E\in \mathrm{In}(T)}\mathcal U_{E}^{(n)}\rightarrow \prod_{S\in{b(T)}}\mathcal Y_S^{(n)}$}
\end{itemize}
such that
\begin{align*}
\lim_{n\rightarrow\infty} n^{-1} \log |\mathcal Y_{S}^{(n)}|= \lim_{n\rightarrow\infty} n^{-1} H(Y_{S}^{(n)})\geq R_S\\
\lim_{n\rightarrow\infty} n^{-1} H(U_{E}^{(n)})\leq \lim_{n\rightarrow\infty} n^{-1} \log |\mathcal U_{E}^{(n)}|\leq C_{E}
\end{align*}
and  for all $T\in\mathcal T$
$$\lim_{n\rightarrow\infty} \mathrm{Pr}(g_T(U_{E}^{(n)}:E\in \mathrm{In}(T))\neq(X_{S}^{(n)}:S \in b(T)))=0.$$
\end{definition}

The closure of the set of all achievable rate tuples for a given network is called the capacity of the network.

\subsection{Functional Dependence Graph}
Given a network $\mathcal G=(\mathcal V,\mathcal E)$ with sets of random variables $\{Y_S:S\in{\mathcal S}\}$ and $\{U_E:E\in\mathcal E\}$ representing the information generated by sources and carried by edges respectively, a valid network code for the network satisfies the following constraints.
\begin{align}
R_{S}&\leq H(Y_{S})\label{eq:rate}\\
H(Y_{\mathcal S})&=\sum_{S\in{\mathcal S}}{H(Y_S)}\label{sourceIndep} \\
H(U_{\mathrm{Out}(S)}|Y_S)&=0, S\in{\mathcal S}\label{sourceEn}\\
H(U_{\mathrm{Out}(V)}|U_{\mathrm{In}(V)})&=0, V\in{\mathcal V\setminus{(\mathcal S\cup\mathcal T)}}\label{edgeEn}\\
%H(U_{Out(S)}|U_{In(s)})&=0, i\in{V\setminus\{S\cup{T}\}} \label{edgeEn}\\
H(Y_{b(T)}|U_{\mathrm{In}(T)})&=0, T\in{\mathcal T}\label{receiverDe}\\
H(U_E)&\leq{C_{E}}, E\in{\mathcal E} \label{edgeCap}
\end{align}
Note that constraint \eqref{sourceIndep} specifies the independence of the sources, \eqref{sourceEn} and \eqref{edgeEn} give encoding constraints, \eqref{receiverDe} corresponds to the decoding requirements, \eqref{edgeCap} means that the entropy of variables carried by any link cannot exceed the link capacity.

Using the encoding and decoding constraints \eqref{sourceEn}-\eqref{receiverDe} for a given network $\mathcal G$ (i.e., local functional dependence constraints at nodes), an FDG can be constructed for the set of source and edge random variables. FDG is defined in \cite{Thakor09} as:

\begin{definition}[Functional dependence graph \cite{Thakor09}]\label{FDG}
Let $\bar{\mathcal V}$ be a set of random variables. A directed graph $\bar{\mathcal G}=(\bar{\mathcal V}, \bar{\mathcal E})$ is called a functional dependence graph for $\bar{\mathcal V}$ if and only if for all $V \in \bar{\mathcal V}$
\begin{equation}\label{dependence}
H(V|\{U:(U,V)\in{\bar{\mathcal E}}\})=0.
\end{equation}
\end{definition}

$\bar{\mathcal G}$ is a directed cyclic graph and of the order $|\bar{\mathcal V}|=|\mathcal E|+|\mathcal S|$. $\bar{\mathcal E}$ is determined by the dominance relationship defined in \eqref{sourceEn}-\eqref{receiverDe}. For $V\in\bar{\mathcal V}$, let $\mathrm{Down}(V)\triangleq\{V':(V,V')\in\bar{\mathcal E}\}$ be the set of immediate downstream nodes (also known as children) of $V$ and $\mathrm{Up}(V)\triangleq\{V':(V',V)\in\bar{\mathcal E}\}$ be the set of immediate upstream nodes (also known as parents) of $V$. Thus, for $S\in\mathcal S$, $\mathrm{Down}(Y_S)=U_{\mathrm{Out}(S)}$ and if $b(T)=S$, $\mathrm{Up}(Y_S)=U_{\mathrm{In(T)}}$. Moreover, for $E\in{\mathcal E}$, $\mathrm{Down}(U_E)=U_{\mathrm{Out}(E)}$ and $\mathrm{Up}(U_E)=U_{\mathrm{In}(E)}$. Provided that the communication network $\mathcal G$ is directed acyclic, every node $Y_S, S \in \mathcal S$ in $\bar{\mathcal G}$ must be a member of a cycle due to the decoding constraints \eqref{receiverDe}.

\section{Main Results}\label{S:RFDG}
\subsection{Reduced Functional Dependence Graph}
For a given network $\mathcal G=(\mathcal V,\mathcal E)$, if it is a connected graph, $|\mathcal E|\geq{|\mathcal V|-1}$. Therefore the order of an FDG of the given network, $|\bar{\mathcal V}|=|\mathcal E|+|\mathcal S|$, is usually much larger than the original network. As the complexity of computing network capacity bounds usually grows exponentially with the size of an FDG, it is desirable to reduce the size of it.

%Capacity outer bounds (and the capacity) for networks are characterized by

Source random variables $Y_S,\ S\in\mathcal S$ may be regarded as primary random variables introducing randomness into the network. Edge random variables $U_{E}, E \in \mathcal E$ may be regarded as secondary random variables since all edge random variables are function of the source random variables. In this section we focus on eliminating certain edge random variables (and hence corresponding nodes in FDGs) such that the capacity of the network remains unchanged.
\begin{definition}
Given a network $\mathcal G=(\mathcal V, \mathcal E)$, let $E=(V,U)\in\mathcal E$, the reduced network $\mathcal G\setminus E=(\mathcal V', \mathcal E')$ is defined by: $\mathcal V'=\mathcal V$ and $\mathcal E'=\mathcal E\setminus E$. Moreover, the information carried by edges in $\mathrm{In}(E)$ are made available to both node $V$ and $U$, i.e. for $K\in\mathrm{In}(E)$, $\mathrm{Head}(K)=\{U,V\}$.
\end{definition}

The procedure for obtaining $\mathcal G\setminus E$ from $\mathcal G$ is shown in Fig.\ref{F:Network}. Note that the reduced network $\mathcal G\setminus E$ is no longer a traditional network since some edges have more than one destination nodes. However, the definition of capacity still applies.

\begin{figure}[htb]
\centering
\includegraphics[scale=0.55]{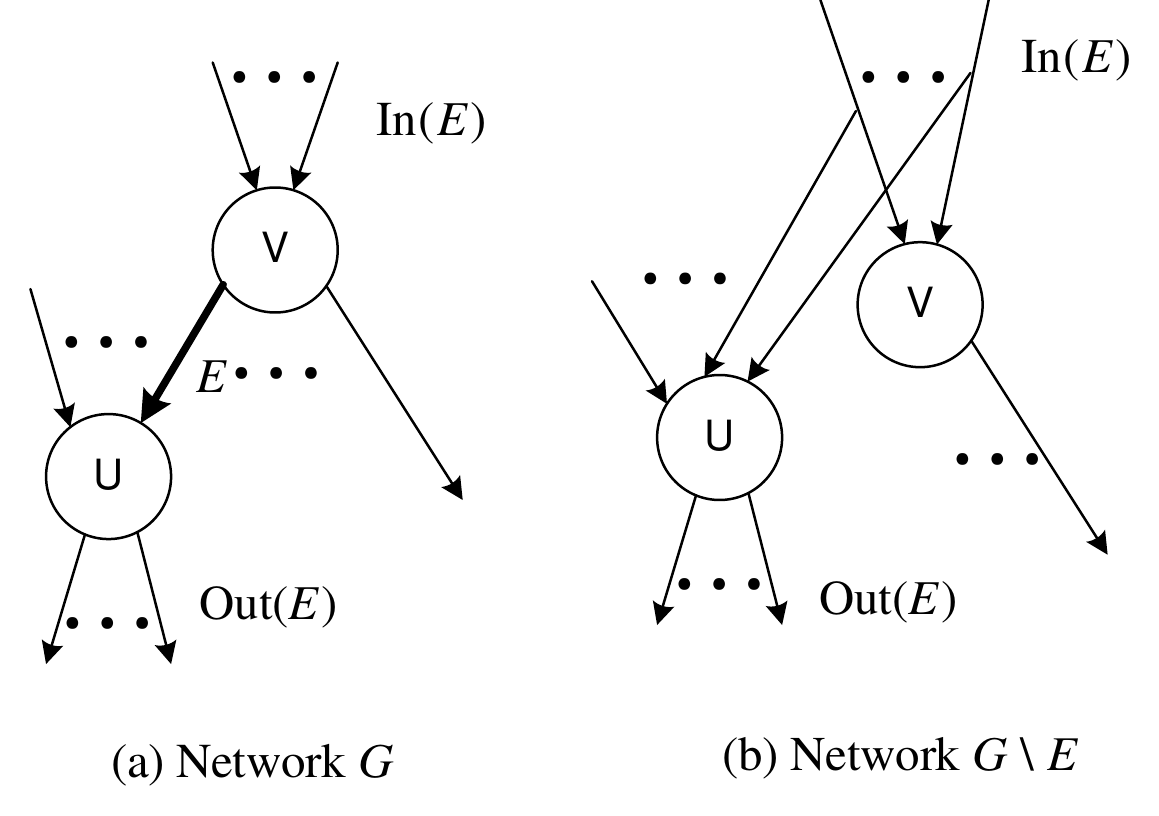}
\caption{Construct new network by removing edge $E$}
\label{F:Network}
\end{figure}

Although the transformation from $\mathcal{G}$ to $\mathcal{G}\setminus E$ makes the network look more complex, it reflects a reduction in their corresponding FDG sub-networks (referred as $\bar{\mathcal{G}}$ and $\bar{\mathcal G}\setminus U_{E}$ respectively) as shown in Fig.\ref{F:FDG}.
\begin{definition}
Given an FDG $\bar{\mathcal G}=(\bar{\mathcal V},\bar{\mathcal E})$, a \emph{reduced} FDG $\bar{\mathcal G}\setminus U_{E}=(\bar{\mathcal V}', \bar{\mathcal E}')$ is given by $\bar{\mathcal V}' \triangleq \bar{\mathcal V}\setminus U_E$ and
%\begin{equation}
%\bar{\mathcal V}' \triangleq \bar{\mathcal V}\setminus U_E
%\end{equation} and
\begin{multline*}
\bar{\mathcal E}'\triangleq \left\{\bar{\mathcal E} \setminus\{\bar E \in \bar{\mathcal E}:\mathrm{Head}(\bar E)=U_{E} \textnormal{ or } \mathrm{Tail}(\bar E)=U_E\}\right\}\\
\cup \{(A,B): A\in\mathrm{Up}(U_E), B\in\mathrm{Down}(U_E)\} .
\end{multline*} %We regard $\bar{\mathcal G}'$ as a reduction of $\bar{\mathcal G}$ as the edge variable $U_E$ is removed.
\end{definition}

\begin{figure}[htb]
\centering
\includegraphics[scale=0.38]{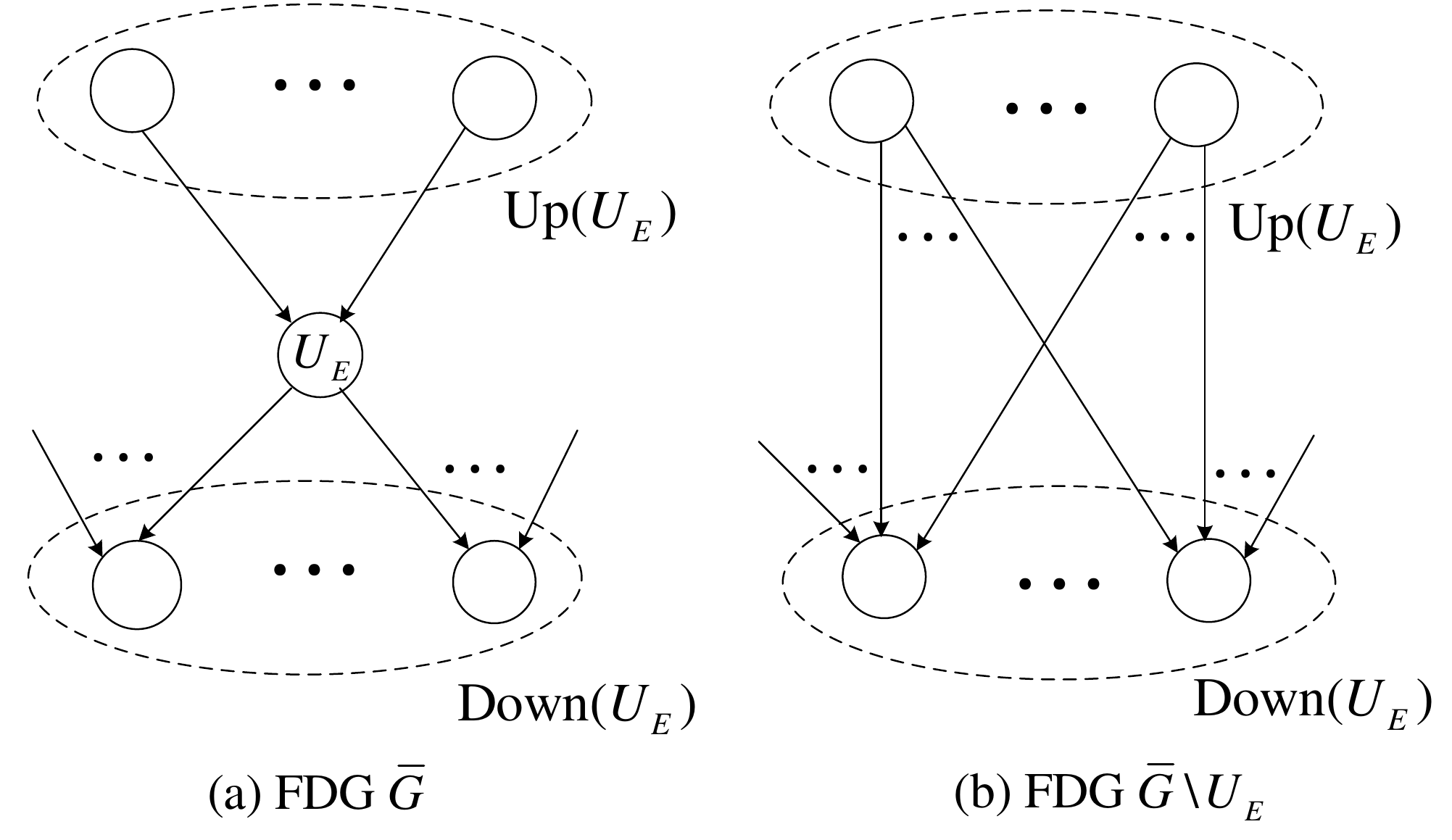}
\caption{Construct new FDG by removing $U_E$}
\label{F:FDG}
\end{figure}

\begin{theorem}\label{reducedNetwork}
Network $\mathcal G$ and $\mathcal G\setminus E$ have the same capacity if the following conditions are satisfied:
\begin{enumerate}
\item{$\mathrm{Tail}(E)\not\in\mathcal S$}
\item{$C_E\geq\sum_{E':E'\in\mathrm{In}(E)}C_{E'}$}
\end{enumerate}
\end{theorem}
\begin{proof}
Assume that an information rate tuple $\mathbf{R}=(R_S:S\in\mathcal S)$ is achievable in the original network $\mathcal G$. According to Definition \ref{NC}, there exists a network code $\textbf{\textsf{C}}$, composed of a set of encoding functions $\{f_E:E\in\mathcal E\}$ and a set of decoding functions $\{g_T:T\in\mathcal T\}$, with which the communication requirement is satisfied. Given $\textbf{\textsf{C}}$, we can easily construct a network code $\textbf{\textsf{C}}'$ for the new network $\mathcal{G}\setminus E$ via the following modification:
\begin{itemize}
\item{excluding the encoding function for the removed edge $E$}
\item{if $\mathrm{Head}(E)\not\in\mathcal T$, for edge $K\in\mathrm{Out}(E)$, let $\mathcal{X}\triangleq\mathrm{In}(K)\setminus E$ be the set of incoming edges to the edge $K$ except $E$. Replacing $f_K$ with a new composite function $f_{K}'=f_{K}(f_E(\prod_{E'\in\mathrm{In}(E)}\mathcal U_{E'}^{(n)}),\prod_{X\in\mathcal X}\mathcal U_{X}^{(n)})$}
\item{if $\mathrm{Head}(E)=T\in\mathcal T$, let $\mathcal{X}=\mathrm{In}(T)\setminus E$. Changing $g_T$ to $g_T'=g_T(f_E(\prod_{E'\in\mathrm{In}(E)}\mathcal U_{E'}^{(n)}),\prod_{X\in\mathcal X}\mathcal U_{X}^{(n)})$ }
\end{itemize}
With above modification, if $\mathrm{Head}(E)\not\in\mathcal T$, $\textbf{\textsf{C}}'=(\textbf{\textsf{C}}\setminus \{f_E,\{f_{K},K\in\mathrm{In}(E)\}\})\cup\{f_{K}'\}$ and if $\mathrm{Head}(E)=T\in\mathcal T$, $\textbf{\textsf{C}}'=(\textbf{\textsf{C}}\setminus \{f_E,g_T\})\cup\{g_T'\}$. Note that the encoding function $f_{K}'$ defines a mapping:$\prod_{X\in\mathcal X}\mathcal U_{X}^{(n)}\cdot\prod_{E'\in\mathrm{In}(E)}\mathcal U_{E'}^{(n)}\rightarrow\mathcal U_{K}^{(n)}$ which produces the same output as $f_K$ in the original network code $\textbf{\textsf{C}}$. Similar arguments hold for the decoding function $g_T'$ if $\mathrm{Head}(E)=T\in\mathcal T$. Therefore, the rate tuple $\mathbf{R}=(R_S:S\in\mathcal S)$ is achievable in $\mathcal{G}\setminus E$ with the network code $\textbf{\textsf{C}}'$.

On the other hand, assume that a rate tuple $\mathbf{R}=(R_S:S\in\mathcal S)$  is achievable in $\mathcal{G}\setminus E$. There exists a network code $\textbf{\textsf{C}}'$, with which the communication requirement can be satisfied in $\mathcal{G}\setminus E$. Based on $\textbf{\textsf{C}}'$, we can construct the network code $\textbf{\textsf{C}}$ for the original network if the conditions stated in this theorem are satisfied.

Condition 1) states $\mathrm{Tail}(E)\not\in\mathcal S$ and Condition 2) states that capacity of edge $E$ is larger than or equal to sum capacity of all edges coming into node $V$. Combined both conditions, we know that edge $E$ is capable of forwarding all received data without any compression. Therefore, we can set the encoding function at edge $E$ as a simple replication, i.e. $\mathcal U_{E}^{(n)}=f_E(\prod_{E'\in\mathrm{In}(E)}\mathcal U_{E'}^{(n)})=\prod_{E'\in\mathrm{In}(E)}\mathcal U_{E'}^{(n)}$. Thus, the network code $\textbf{\textsf{C}}=\textbf{\textsf{C}}'\cup f_E$. Note that this construction does not violate any encoding constraint for downstream edges of $E$ since the same information is available to them for encoding as in the reduced FDG. Therefore, all the variables in the original network are the same as that in the reduced network and thus decoding requirements are satisfied.
\end{proof}

Due to the correspondence of network and its FDG, Theorem \ref{reducedNetwork} can be applied to reduce the size of FDG. Before introducing the corollary of Theorem \ref{reducedNetwork} in FDG, the definition of redundant random variable is made.

\begin{definition}\label{newFDG}
A node $U_E$ (corresponding to the edge random variable $U_E$) in a functional dependence graph $\bar{\mathcal G}=(\bar{\mathcal V}, \bar{\mathcal E})$ is called \emph{redundant} if the capacity region of the reduced FDG $\bar{\mathcal G}\setminus U_E$ is the same as the original FDG $\bar{\mathcal G}$.
\end{definition}

\begin{corollary}\label{singleNode}
A node $U_E\in\big(\bar{\mathcal V}\setminus{Y_{\mathcal S}}\big)$ is redundant if the following conditions are satisfied:
\begin{enumerate}
\item{$Y_{\mathcal S}\cap\mathrm{Up}(U_E)=\emptyset$}
\item{$C_E\geq \sum_{A:U_A\in\mathrm{Up}(U_E)}C_A$}
\end{enumerate}
\end{corollary}

%\begin{proof}
Corollary \ref{singleNode} follows from Theorem \ref{reducedNetwork} and Definition \ref{newFDG}.
%\end{proof}
A simple extension of Corollary \ref{singleNode} from single edge variable to a group of edge variables gives the following corollary.
\begin{corollary}\label{groupNode}
A set of nodes $U_{\mathcal X}\subset\big(\bar{\mathcal V}\setminus Y_{\mathcal S}\big)$ is redundant if the following conditions are satisfied:
\begin{enumerate}
\item{$\mathrm{Up}(U_{E_1})=\mathrm{Up}(U_{E_2}),\ \forall E_1,E_2\in{\mathcal X}$}
\item{$\mathrm{Down}(U_{E_1})=\mathrm{Down}(U_{E_2}),\ \forall E_1,E_2\in{\mathcal X}$}
\item{$Y_{\mathcal S}\cap\mathrm{Up}(U_{\mathcal X})=\emptyset$ and $\sum_{E\in\mathcal X}C_E\geq\sum_{U_A\in\mathrm{Up}(U_{\mathcal X})}C_A$}
\end{enumerate}
\end{corollary}
%\begin{proof}
The proof of Corollary \ref{groupNode} follows %directly from Corollary \ref{singleNode}
by treating the set of variables $U_{\mathcal X}$ as a single ``super-variable''.
%\end{proof}
Note that the redundant variables identified by Corollary \ref{singleNode} and Corollary \ref{groupNode} correspond to the edges in original network where network coding is not necessary since they have sufficient capacity to forward all received information to their children.

In practise, most networks use packet-based transmissions. Therefore, we can assume unit-edge capacity. The edges that have large capacity (convey multiple packets) can be represented by multiple parallel edges of unit capacity. Based on this assumption, Corollary 1 and Corollary 2 can be simplified as follows.
\begin{corollary}\label{unitSingle}
For $U_E\in\big(\bar{\mathcal V}\setminus Y_{\mathcal S}\big)$, it is redundant if $|\mathrm{Up}(U_E)|=1$ and $\mathrm{Up}(U_E)\cap Y_{\mathcal S}=\emptyset$
\end{corollary}

\begin{corollary}\label{unitGroup}
Let $U_{\mathcal X}\subset\big(\bar{\mathcal V}\setminus Y_{\mathcal S}\big)$,  they are redundant if the following conditions are satisfied.
\begin{enumerate}
\item{$\mathrm{Up}(U_{E_1})=\mathrm{Up}(U_{E_2}),\ \forall E_1,E_2\in{\mathcal X}$}
\item{$\mathrm{Down}(U_{E_1})=\mathrm{Down}(U_{E_2}),\ \forall E_1,E_2\in{\mathcal X}$}
\item{$\mathrm{Up}(U_{\mathcal X})\cap Y_{\mathcal S}=\emptyset$ and $|\mathrm{Up}(U_{\mathcal X})|\leq |U_{\mathcal X}|$}
\end{enumerate}
\end{corollary}

\begin{example}
An example of reduced FDG is shown in Fig.\ref{F:S2Unicast}. The network in Fig.\ref{F:S2Unicast}(a) is constructed in \cite{Kamath11} and each edge in the network has unit capacity. Fig.\ref{F:S2Unicast}(c) gives the FDG constructed based on Definition \ref{FDG}. This FDG can be reduced to the one shown in Fig.\ref{F:S2Unicast}(b) based on Corollary \ref{unitSingle}.
\end{example}

\begin{figure}[htb]
\centering
\includegraphics[scale=0.47]{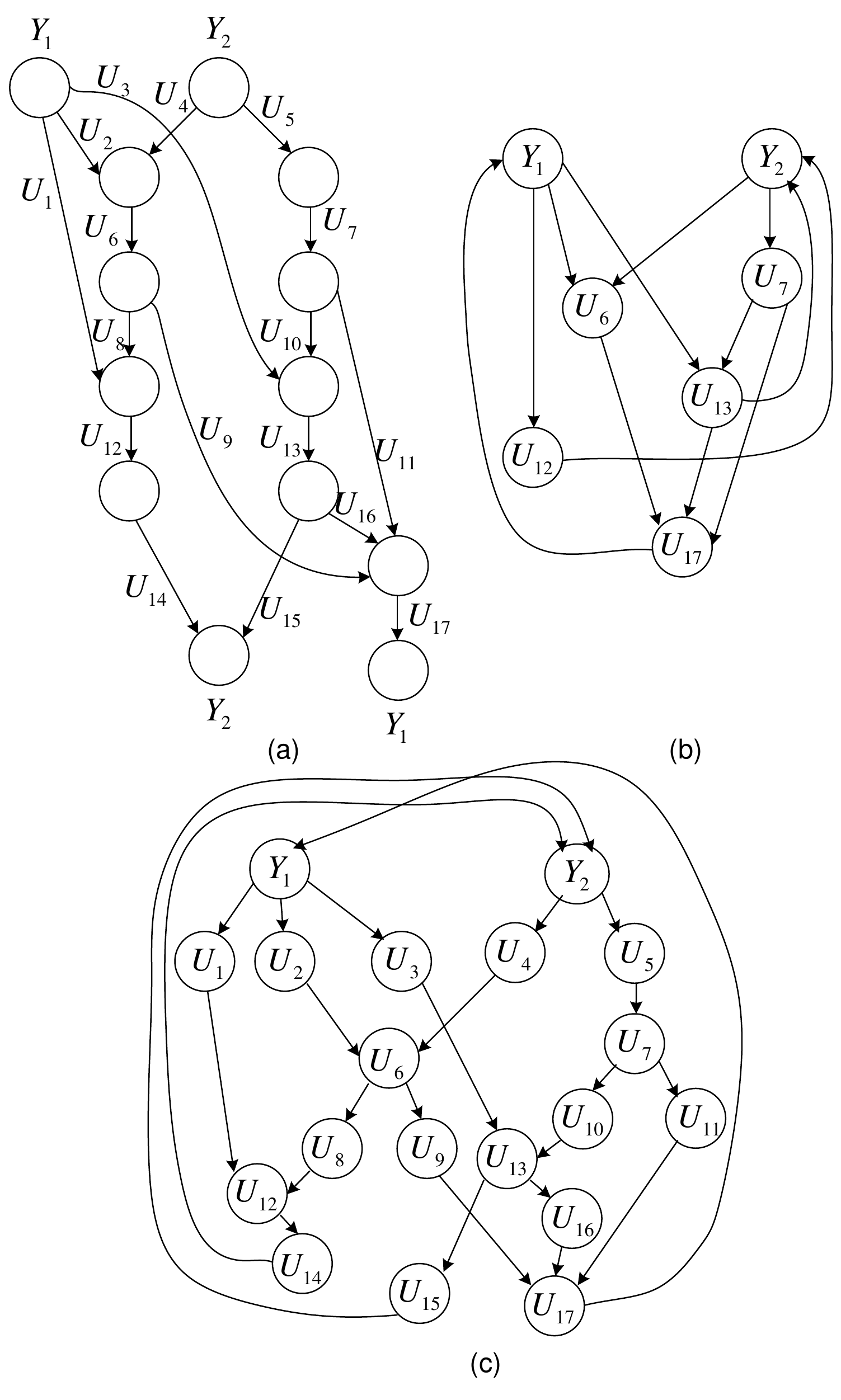}%[width=3in, height=3.8in]
\caption{(a)A network in \cite{Kamath11} (b)The Reduced FDG (proposed) (c)The conventional FDG}
\label{F:S2Unicast}
\end{figure}

\subsection{Further Reduction of FDG}
It has been proved in \cite{Dougherty05} that linear network coding is not sufficient to achieve the capacity for general networks. However, in most practical scenarios, linear network coding is preferred due to its simplicity. With linear network coding, the data carried by an edge $E$ is a linear combination of input data available at $\mathrm{Head}(E)$. Without loss of generality, we may assume that each edge has unit capacity and the encoding function at edge $E$ is described by\cite{Koetter03}:
\begin{equation*}\label{LinearNC}
f_E: U_E=\sum_{E':E'\in\mathrm{In}(E)}\varepsilon_{E,E'}U_{E'}
\end{equation*}
where the coefficients $\varepsilon_{E,E'}$ are elements chosen from designed alphabet and $H(U_E)\leq 1,\forall E\in \mathcal E$.

\begin{theorem}\label{LinearSingleNode}
Assume each edge in the network has unit capacity, then the network $\mathcal G$ and $\mathcal G\setminus E$ (refer to Fig.\ref{F:Network}) have the same \emph{linear network coding capacity} if one of the following conditions is satisfied:
\begin{enumerate}
\item{$\mathrm{Tail}(E)\not\in\mathcal S$ and $|\mathrm{In}(E)|=1$}
\item{$\mathrm{Head}(E)\not\in\mathcal T$ and $|\mathrm{Out}(E)|=1$}
\end{enumerate}
\end{theorem}
\begin{proof}
As shown in the proof of Theorem \ref{singleNode}, if a rate tuple $\mathbf{R}=(R_S:S\in\mathcal S)$ is achievable in the $\mathcal G$, it is always achievable in $\mathcal G\setminus E$ with the same linear network code after simple modification.

On the other hand, if a rate tuple $\mathbf{R}$ is achieved in $\mathcal G\setminus E$ with a linear network code $\textbf{\textsf{C}}'$, we will show that a linear network code $\textbf{\textsf{C}}$ can be constructed for the original network $\mathcal G$ with given conditions.

For $K\in\mathrm{Out}(E)$, let $\mathcal X=\mathrm{In}(K)\setminus E$. Since $\textbf{\textsf{C}}'$ is a linear network code, the encoding function for edge $K$ is of the form: $f_K:U_K=\sum_{E:E\in\mathrm{In}(K)}\varepsilon_{E,K}U_E=\sum_{E'\in\mathrm{In}(E)}\varepsilon_{E',K}U_{E'}+\sum_{X\in\mathcal X}\varepsilon_{X,K}U_X$.  If condition 1) is satisfied, $\textbf{\textsf{C}}$ is constructed from $\textbf{\textsf{C}}'$ by letting edge $E$ perform simple forwarding, i.e.$f_E:U_E=U_{\mathrm{In}(E)}$. If condition 2) is satisfied, the encoding for edge $E$ can be set to $f_E:U_E=\sum_{E'\in\mathrm{In}(E)}\varepsilon_{E',K}U_{E'}$. For the edge $K=\mathrm{Out}(E)$, the encoding function is changed to $f_K:U_K=U_E+\sum_{X\in\mathcal X}\varepsilon_{X,K}U_X$. For either case, the constructed code $\textbf{\textsf{C}}$ achieves the same rate tuple $\mathbf{R}=(R_S:S\in\mathcal S)$ since all edges carry the same information as that in $\mathcal G\setminus E$.
\end{proof}
%\begin{figure}[htb]
%\centering
%\includegraphics[scale=0.55]{LinearNetworkChange}
%\caption{Construct new network by removing the edge $E$}
%\label{F:LinearNetworkChange}
%\end{figure}

Applying Theorem \ref{LinearSingleNode} in FDG results further reduction described by the following corollary.
\begin{corollary}\label{FDGLinearSingle}
When linear network coding capacity is interested, a node $U_E\in\big(\bar{\mathcal V}\setminus Y_{\mathcal S}\big)$ is redundant if one of the following conditions is satisfied:
\begin{enumerate}
\item{$|\mathrm{Up}(U_E)|=1$ and $Y_{\mathcal S}\cap\mathrm{Up}(U_E)=\emptyset$}
\item{$|\mathrm{Down}(U_E)|=1$ and $Y_{\mathcal S}\cap\mathrm{Down}(U_E)=\emptyset$}
\end{enumerate}
\end{corollary}

Note that the extension to group nodes is straightforward and thus omitted here due to space limitation.
\begin{example}
Consider the network shown in Fig.\ref{F:2unicast}(a) which appeared in \cite{Song11} and has unit edge capacity. The corresponding FDG of this network is shown in Fig.\ref{F:2unicast}(b).  By Corollary \ref{singleNode}, it is reduced to \ref{F:2unicast}(c) and it can be further reduced to the one in Fig.\ref{F:2unicast}(d) according to Corollary \ref{FDGLinearSingle} if only linear network coding capacity is interested. With the reduced FDG, it is easy to show that $H(Y_1,Y_2|U_8)=0$. Thus the sum capacity that can be achieved with linear network coding is bounded by: $R_1+R_2\leq H(U_8)\leq 1$, which is the same as the capacity achieved by routing.
\end{example}

\begin{figure}[htb]
\centering
\includegraphics[scale=0.42]{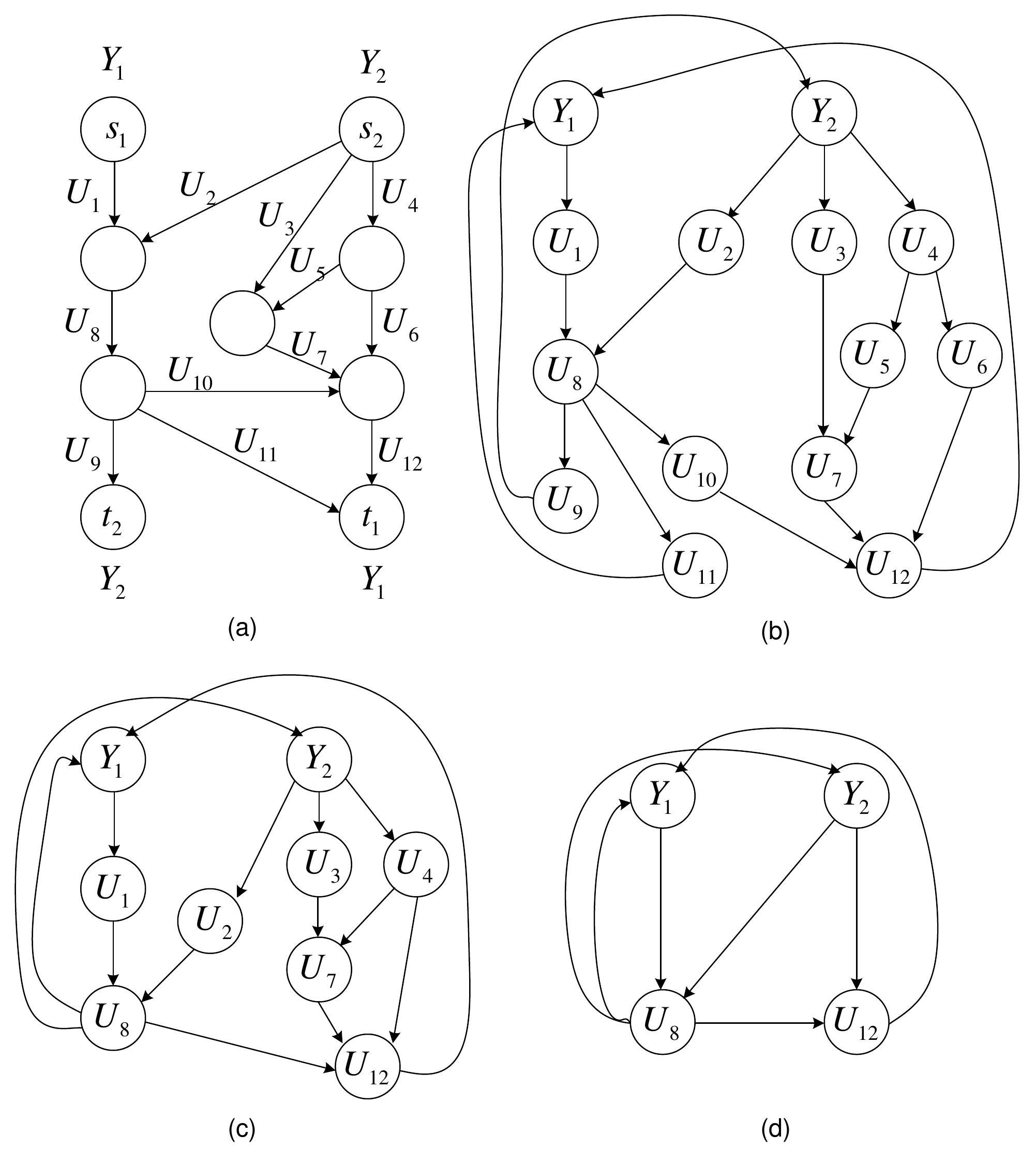}
\caption{(a)A network taken from \cite{Song11} (b)The conventional FDG (c)The reduced FDG (proposed) (d)The further reduced FDG}
\label{F:2unicast}
\end{figure}

\section{Applications of Reduced FDG}
Although the concept of reduced FDG is very simple and the procedures to obtain reduced FDG are straightforward, the range of its application is wide and two of them are discussed in this section.
\subsection{Complexity Reduction for Computing LP Bound}\label{S:LPbound}
The Linear Programming (LP) bound is the set of all rate tuples satisfying \eqref{eq:rate}-\eqref{edgeCap} together with the basic inequalities. The basic inequalities are non-negativity of entropy, conditional entropy and conditional mutual information. The set of the elemental basic inequalities
\begin{align}
H(A|{\{Y_{\mathcal S},U_{\mathcal E}\}\setminus{A}})\geq{0},&\: A\in\{Y_{\mathcal S},U_{\mathcal E}\}\label{Shannon1}\\
I(A;B|\mathcal C)\geq{0},&\: A\neq{B}\neq \emptyset, A,B \in \{Y_{\mathcal S},U_{\mathcal E}\},\nonumber\\
 &\: \mathcal C \subseteq \{Y_{\mathcal S},U_{\mathcal E}\}\setminus \{A,B\}\label{Shannon2}
\end{align}
represents (implies) all Shannon-type inequalities for the random variables $Y_{\mathcal S},U_{\mathcal E}$ and is minimal \cite{Yeung08}. Note that the constraints \eqref{eq:rate}-\eqref{Shannon2} are linear and hence the LP bound can be computed by solving the linear program
\begin{equation}\label{eq:LPBound}
%  \begin{split}
    \text{max }\sum_{S\in \mathcal S} w_S H(Y_S) \text{ subject to }\eqref{sourceIndep}-\eqref{edgeCap},\eqref{Shannon1}-\eqref{Shannon2}
%  \end{split}
\end{equation}
where $w_S$ is any non-negative constant for source $S$ called weight coefficient.

For every set of chosen weight, the dimension of this optimization problem is $2^N-1$, where $N=|\mathcal E|+|\mathcal S|$ and the number of constraints is $N+{{N}\choose{2}}2^{N-2}+1+2|\mathcal E|+|\mathcal T|$, including $N+{{N}\choose{2}}2^{N-2}$ elemental information inequalities, 1 equality representing source independence, $|\mathcal E|$ equalities for encoding requirements, another $|\mathcal E|$ inequalities capturing the edge capacity constraints and $|\mathcal T|$ equalities stating the decoding requirements.

Note that the dimension of this optimization problem grows exponentially with the order of its FDG, $N$. Therefore, any reduction of FDG size results in exponential reduction of the problem dimension. The reduction in the number of constraints is more significant since it grows even faster with the size of FDG.

\begin{example}
(Butterfly Network):\textnormal{Consider the Butterfly network shown in Fig.\ref{Butterfly}(a). If the original FDG shown in Fig.\ref{Butterfly}(b) is used. The number of variables involved is $N=9$ and the dimension of the linear program is $511$. A matrix representing the constraint set is of size $4634\times 511$. With the reduced FDG shown in Fig.\ref{Butterfly}(c), $N$ reduces to $5$, thus the problem dimension and the constraint size reduces to $31$ and $91\times 31$ respectively. The overall complexity reduction exceeds $99\%$. Combining the reduced FDG with other complexity reduction techniques shown in \cite{Thakor11}, the dimension of the LP bound for Butterfly network can be further reduced to $8$.}
\end{example}

\begin{figure}[htb]
\centering
\includegraphics[scale=0.44]{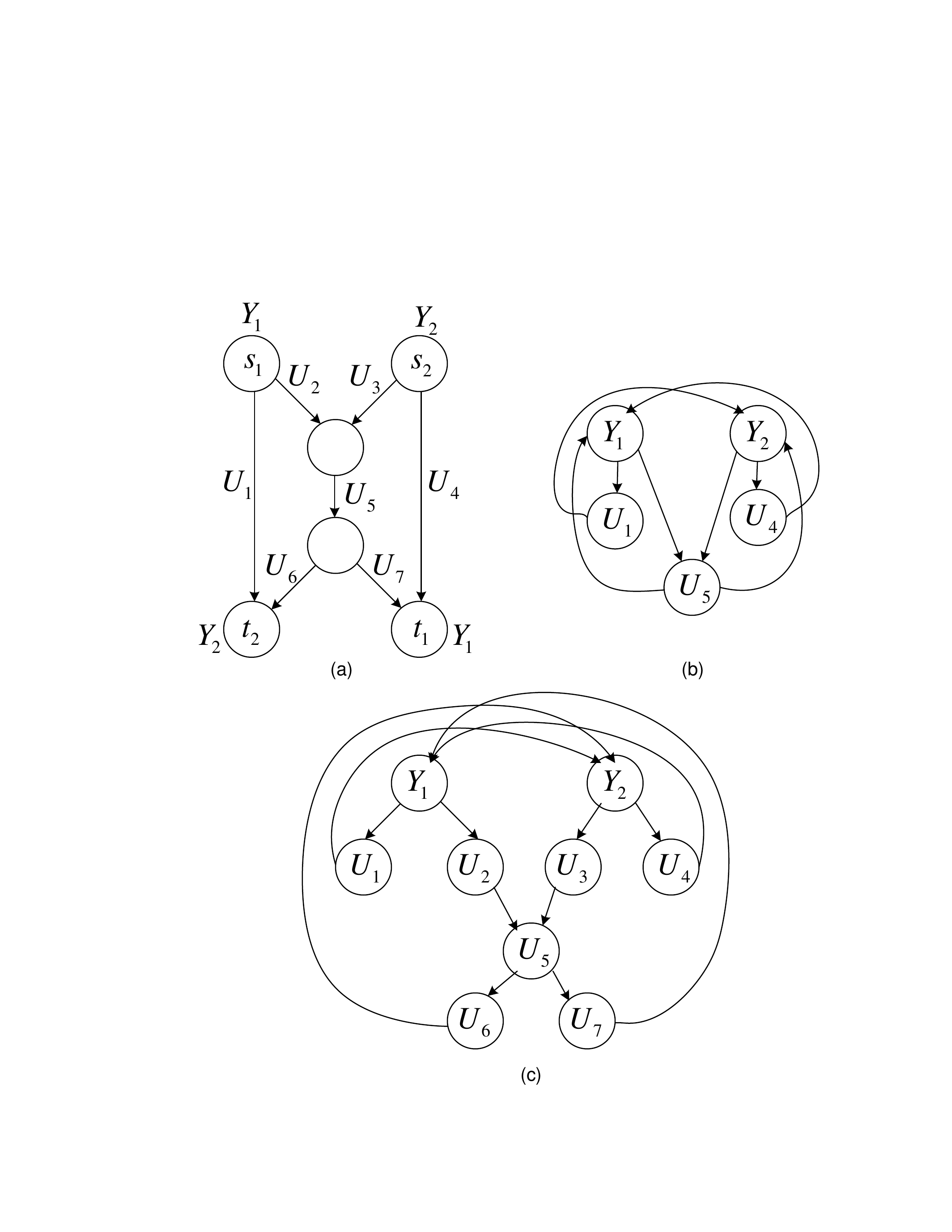}
\caption{(a)Butterfly Network (b)The reduced FDG (c)The conventional FDG}
\label{Butterfly}
\end{figure}

\subsection{Complexity Reduction in Algebraic Formulation of the Scalar Network Coding}\label{S:ImprovedFDG}
\begin{figure*}
%\hline
\begin{equation}\label{eq:long}
\begin{aligned}
M=A(I-F)^{-1}B=\left[ \begin{array}{ccc}
\alpha_2\beta_2\gamma_5 & \alpha_2\beta_1\gamma_3+\alpha_2\beta_2\gamma_4 & \alpha_1\gamma_1+\alpha_2\beta_1\gamma_2\\
\alpha_4\gamma_6+\alpha_3\beta_2\gamma_5 & (\alpha_3\beta_1+\alpha_4\beta_3)\gamma_3+\alpha_3\beta_2\gamma_4 & (\alpha_3\beta_1+\alpha_4\beta_3)\gamma_2\\
\alpha_5\gamma_6+\alpha_6\gamma_5 & \alpha_5\beta_3\gamma_3+\alpha_6\gamma_4 & \alpha_5\beta_3\gamma_2\\
\end{array}
\right]
\end{aligned}
\end{equation}
\hrulefill
\end{figure*}
The directed line graph is introduced in \cite{Koetter03} to construct an algebraic formulation of the scalar network coding. For a network $\mathcal G=(\mathcal V,\mathcal E)$, let its line graph be denoted by $\tilde{\mathcal G}=(\tilde{\mathcal V},\tilde{\mathcal E})$ and its FDG be denoted by $\bar{\mathcal G}=(\bar{\mathcal V},\bar{\mathcal E})$. Each $V\in \tilde{\mathcal V}$ corresponds to an edge $E\in \mathcal E$ and $(V_1,V_2)\in{\tilde{\mathcal E}}$ if and only if $\mathrm{Head}(V_1)=\mathrm{Tail}(V_2)$ in $\mathcal G$. Compare it with the definition of FDG, it is not difficult to conclude that $\tilde{\mathcal G}$ is a subgraph of $ \bar{\mathcal G}$. Therefore, $\bar{\mathcal G}$ can be used to construct the algebraic formulation as well.

Following the similar notation used in \cite{Koetter03}, let $A$ be the matrix that maps the source to its out-edges, $F$ be adjacent matrix of $\tilde{\mathcal G}$ and $B$ be the matrix that maps in-edges of the receiver to decoded symbols. The transition matrix of the network is formulated as $M=A(I-F)^{-1}B$ \cite{Koetter03}. The total number of variables involved in this formulation is $|\bar{\mathcal E}|$ and adjacent matrix $F$ is of the size $|\bar{\mathcal V}|\times{|\bar{\mathcal V}|}$. Let $\Delta_V$ and $\Delta_E$ denote the amount of nodes and edges removed from the original FDG respectively. Therefore, with the reduced FDG, the variables involved in the formulation is reduced by $\Delta_E$ and the reduction in complexity for computing the transition matrix is  proportional to $\Delta_V$.
\begin{example}
Fig.\ref{binaryNetwork}(a) shows a network taken from \cite{Dougherty05} that is solvable only over fields with characteristic 2 and Fig.\ref{binaryNetwork}(b) shows the corresponding reduced FDG. By labeling each edge in the reduced FDG, we can formulate the scalar linear network coding as:

\small{
\begin{equation*}
A= \left[ \begin{array}{ccccc}
\alpha_1 & \alpha_2 & 0 & 0 & 0\\
0 & \alpha_3 & \alpha_4 & 0 & 0\\
0 & 0 & \alpha_5 & 0 & \alpha_6
\end{array} \right]
%B^T= \left[ \begin{array}{ccccc}
%0 & 0 & \gamma_6 & 0 & \gamma_5\\
%0 & 0 & 0 & \gamma_3 & \gamma_4\\
%\gamma_1 & 0 & 0 & \gamma_2 & 0\\
%\end{array} \right]
\end{equation*}
\begin{equation*}
B^T= \left[ \begin{array}{ccccc}
0 & 0 & \gamma_6 & 0 & \gamma_5\\
0 & 0 & 0 & \gamma_3 & \gamma_4\\
\gamma_1 & 0 & 0 & \gamma_2 & 0\\
\end{array} \right],
F= \left[ \begin{array}{ccccc}
0 & 0 & 0 & 0 & 0\\
0 & 0 & 0 & \beta_1 & \beta_2\\
0 & 0 & 0 & \beta_3 & 0\\
0 & 0 & 0 & 0 & 0\\
0 & 0 & 0 & 0 & 0\\
\end{array} \right]\qquad
\end{equation*}
}

\normalsize

%F= \left[ \begin{array}{ccccc}
%0 & 0 & 0 & 0 & 0\\
%0 & 0 & 0 & \beta_1 & \beta_2\\
%0 & 0 & 0 & \beta_3 & 0\\
%0 & 0 & 0 & 0 & 0\\
%0 & 0 & 0 & 0 & 0\\
%\end{array} \right]\qquad\\

The transition matrix is given in \eqref{eq:long}.
\textnormal{By setting the diagonal element to be 1 and off-diagonal element to be 0, it can easily be proved that solution exists only if all the variables are chosen from a field with characteristic 2. Note that the total number of variables used in this formulation is 15 and the size of $F$ is only $5\times 5$. However, if the original line graph is used, the number of variables is 28 and the dimension of the adjacent matrix is $18\times 18$. Therefore, the total amount of reduction in problem dimension is $46\%$ and in complexity is $92\%$.}
\end{example}

\begin{figure}[htb]
\centering
\includegraphics[scale=0.42]{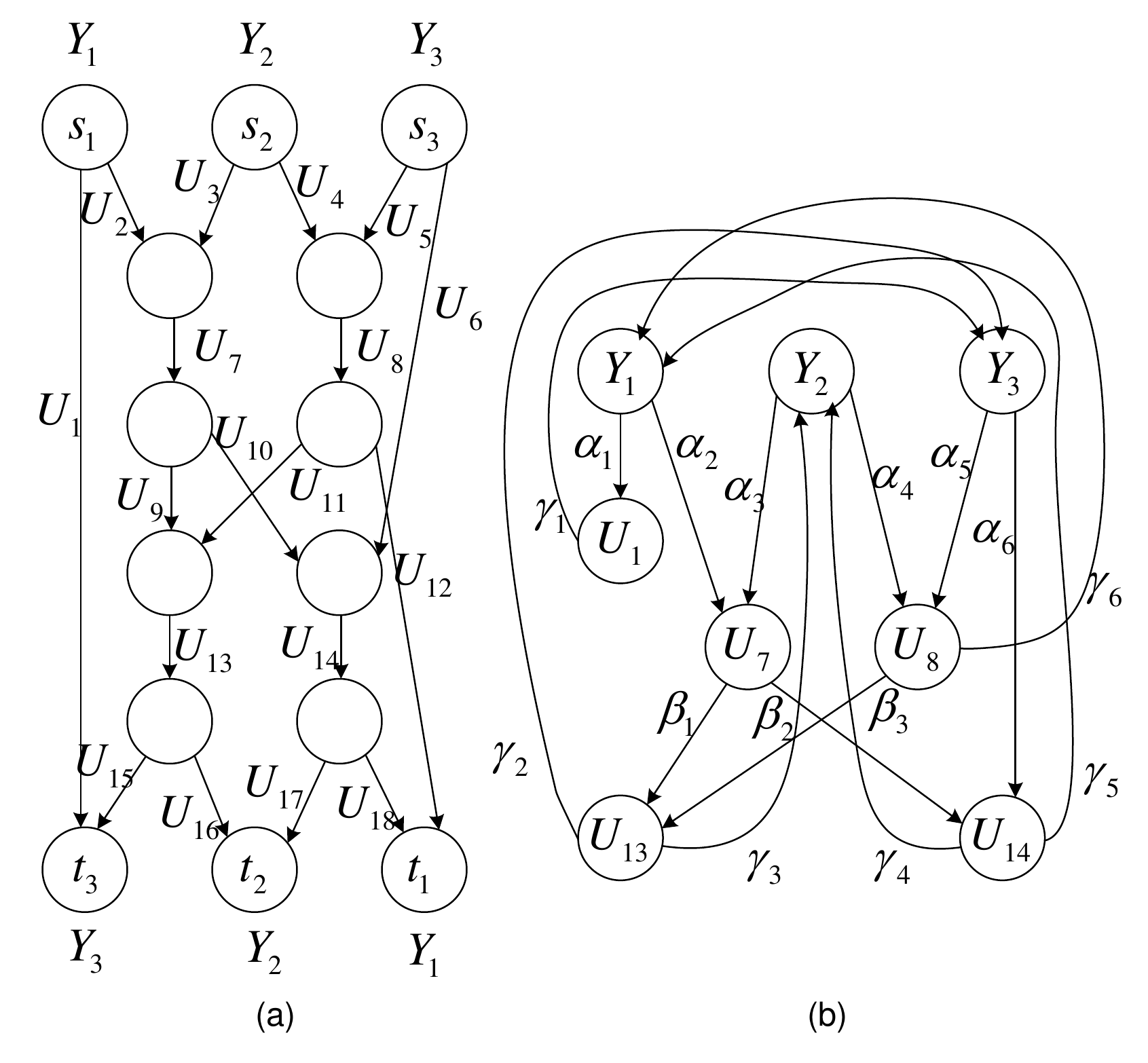}
\caption{A network that is solvable only over fields with characteristic 2}
\label{binaryNetwork}
\end{figure}

\section{Conclusions}
In this paper, reduced FDG is obtained from the original FDG by removing all the edges where network coding cannot be performed and where network coding is not necessary. With reduced FDG, various capacity bounds of the network can be computed much more efficiently without any loss in the tightness. Two applications of reduced FDG are discussed in this paper, one is in the computation of LP bound and the other is in algebraic formulation of scalar linear network coding. In both cases, the reduced FDG significantly reduces the problem size and complexity. Note that the applications of the reduced FDG is not limited to the areas discussed in this paper. It may also help to identify the encoding complexity and simplify the code construction.

\section*{Acknowledgement}
The work of Xiaoli Xu and Yong Liang Guan was  supported by the Advanced Communications Research Program DSOCL06271, a research grant from the Directorate of Research and Technology (DRTech), Ministry of Defence, Singapore. The work of Satyajit Thakor was partially supported by a grant from the University Grants Committee of the Hong Kong Special Administrative Region, China (Project No. AoE/E-02/08).

\end{document}